\documentclass[runningheads]{llncs}

\usepackage{cite}

\usepackage{url}
\usepackage{graphicx}

\usepackage{float}
\usepackage{amssymb}
\usepackage{amsmath}
\usepackage{xfrac}
\usepackage{mathtools}
\usepackage{enumitem}

\usepackage[bookmarks=true, citecolor=black, linkcolor=black, colorlinks=true]{hyperref}

\usepackage{tikz}
\usetikzlibrary{calc}
\usepackage{subcaption}
\usetikzlibrary{patterns}

\usepackage{todonotes}

\usepackage{enumitem}

\usetikzlibrary{hobby}
\usetikzlibrary{shapes}

\newcommand{\var}{var}
\newcommand{\tw}{tw}

\newcommand{\bw}{bw}

\newcommand{\search}{\textup{Search}}
\newcommand{\searchvx}{\textup{SearchVertex}}

\newcommand{\lit}{lit}
\newcommand{\sat}{sat}

\begin{document}

\title{Characterizing Tseitin-formulas with short regular resolution refutations\thanks{This work has been partly supported by the PING/ACK project of the French National Agency for Research (ANR-18-CE40-0011).}}

\author{
	Alexis de Colnet \and Stefan Mengel 
}

\institute{
CNRS, UMR 8188, Centre de Recherche en Informatique de Lens (CRIL), Lens, F-62300, France\\
Univ. Artois, UMR 8188, Lens, F-62300, France
}

\maketitle

\begin{abstract}
\,\,Tseitin-formulas are systems of parity constraints whose structure is described by a graph. These formulas have been studied extensively in proof complexity as hard instances in many proof systems. In this paper, we prove that a class of unsatisfiable Tseitin-formulas of bounded degree has regular resolution refutations of polynomial length if and only if the treewidth of all underlying graphs~$G$ for that class is in $O(\log|V(G)|)$. To do so, we show that any regular resolution refutation of an unsatisfiable Tseitin-formula with graph $G$ of bounded degree has length $2^{\Omega(\tw(G))}/|V(G)|$, thus essentially matching the known $2^{O(\tw(G))}\textup{poly}(|V(G)|)$ upper bound up. Our proof first connects the length of regular resolution refutations of unsatisfiable Tseitin-formulas to the size of representations of \emph{satisfiable} Tseitin-formulas in decomposable negation normal form (DNNF). Then we prove that for every graph~$G$ of bounded degree, every DNNF-representation of every satisfiable Tseitin-formula with graph $G$ must have size $2^{\Omega(\tw(G))}$ which yields our lower bound for regular resolution. 

\end{abstract}

\newcommand{\decisionNode}{
\raisebox{-.5\height}{\begin{tikzpicture}
\def\x{1.2};
\def\y{0.6};
\node[draw,ellipse,inner sep=1, font=\small] (d) at (0.5*\x,1*\y) {$x$};
\node[inner sep=1, font=\small] (c0) at (0.1*\x,0*\y) {$c_0$};
\node[inner sep=1, font=\small]  (c1) at (0.9*\x,0*\y) {$c_1$};

\draw[-stealth, densely dashed] (d) -- (c0);
\draw[-stealth] (d) -- (c1);
\end{tikzpicture}}
}

\newcommand{\decisionNodeProof}{
\raisebox{-.5\height}{\begin{tikzpicture}
\def\x{1.2};
\def\y{0.6};
\node[draw,ellipse,inner sep=1, font=\small] (d) at (0.5*\x,1*\y) {$x_e$};
\node[inner sep=1, font=\small] (c0) at (0*\x,0*\y) {$u_{k_0}$};
\node[inner sep=1, font=\small]  (c1) at (1*\x,0*\y) {$u_{k_1}$};

\draw[-stealth, densely dashed] (d) -- (c0);
\draw[-stealth] (d) -- (c1);
\end{tikzpicture}}
}

\newcommand{\decisionNodeSk}{
\raisebox{-.5\height}{
\begin{tikzpicture}
\def\x{1.2};
\def\y{0.6};
\node[draw,ellipse,inner sep=1, font=\small] (d) at (0.5*\x,1*\y) {$x_e$};
\node[inner sep=1, font=\small] (c0) at (0*\x,0*\y) {$s^0_v$};
\node[inner sep=1, font=\small]  (c1) at (1*\x,0*\y) {$s^1_v$};

\draw[-stealth, densely dashed] (d) -- (c0);
\draw[-stealth] (d) -- (c1);
\end{tikzpicture}
}
}

\newcommand{\decisionNodeAndSk}{
\raisebox{-.5\height}{
\begin{tikzpicture}
\def\x{1.2};
\def\y{0.6};
\node[draw,ellipse,inner sep=1, font=\small] (d) at (0.45*\x,1*\y) {$x_e$};
\node[inner sep=1, font=\small] (c0) at (0*\x,0*\y) {$\land$};
\node[inner sep=1, font=\small]  (c1) at (0.9*\x,0*\y) {0};

\node[inner sep=1, font=\small] (l) at (-0.25*\x,-.8*\y) {$s^i_v$};
\node[inner sep=1, font=\small] (r) at (0.25*\x,-0.8*\y) {$s^{1-i}_b$};

\draw[-stealth, densely dashed] (d) -- (c0);
\draw[-stealth] (d) -- (c1);

\draw (c0) -- (l);
\draw (c0) -- (r);
\end{tikzpicture}
}
}

\newcommand{\decisionNodeAndZk}{
\raisebox{-.5\height}{
\begin{tikzpicture}
\def\x{1.2};
\def\y{0.6};
\node[draw,ellipse,inner sep=1, font=\small] (d) at (0.45*\x,1*\y) {$x_e$};
\node[inner sep=1, font=\small] (c0) at (0*\x,0*\y) {0};
\node[inner sep=1, font=\small]  (c1) at (0.9*\x,0*\y) {$\land$};

\node[inner sep=1, font=\small] (l) at (0.68*\x,-.8*\y) {$s^i_v$};
\node[inner sep=1, font=\small] (r) at (1.15*\x,-0.8*\y) {$s^{1-i}_b$};

\draw[-stealth, densely dashed] (d) -- (c0);
\draw[-stealth] (d) -- (c1);

\draw (c1) -- (l);
\draw (c1) -- (r);
\end{tikzpicture}
}
}

\newcommand{\decisionNodeDNNF}{
\raisebox{-.5\height}{\begin{tikzpicture}
\def\x{1.5};
\def\y{0.75};
\node[inner sep=1, font=\scriptsize] (o) at (0.5*\x,1*\y) {$\lor$};
\node[inner sep=1, font=\scriptsize] (a0) at (0.2*\x,0.6*\y) {$\land$};
\node[inner sep=1, font=\scriptsize] (a1) at (0.8*\x,0.6*\y) {$\land$};
\node[inner sep=1, font=\small] (l0) at (0.05*\x,0.1*\y) {$\overline{x}$};
\node[inner sep=1, font=\small]  (l1) at (0.65*\x,0.05*\y) {$x$};
\node[inner sep=1, font=\small] (c0) at (0.35*\x,0.05*\y) {$c_0$};
\node[inner sep=1, font=\small]  (c1) at (0.95*\x,0.05*\y) {$c_1$};

\draw (o) -- (a0);
\draw (o) -- (a1);

\draw (a0) -- (l0);
\draw (a0) -- (c0);

\draw (a1) -- (l1);
\draw (a1) -- (c1);
\end{tikzpicture}}
}

\newcommand{\guessingNode}{
\raisebox{-.5\height}{\begin{tikzpicture}
\def\x{1.2};
\def\y{0.6};
\node[draw,ellipse,inner sep=1, font=\scriptsize] (d) at (0.5*\x,1*\y) {$\hspace{1pt}?\hspace{1pt}$};
\node[inner sep=1, font=\small] (cl) at (0*\x,0*\y) {$c_l$};
\node[inner sep=1, font=\small]  (cr) at (1*\x,0*\y) {$c_r$};

\draw[-stealth] (d) -- (cl);
\draw[-stealth] (d) -- (cr);
\end{tikzpicture}}
}

\newcommand{\guessingNodeDNNF}{
\raisebox{-.5\height}{\begin{tikzpicture}
\def\x{1};
\def\y{0.5};
\node[inner sep=1, font=\scriptsize] (o) at (0.5*\x,1*\y) {$\lor$};
\node[inner sep=1, font=\small] (cl) at (0*\x,0*\y) {$c_l$};
\node[inner sep=1, font=\small]  (cr) at (1*\x,0*\y) {$c_r$};

\draw (o) -- (cl);
\draw (o) -- (cr);
\end{tikzpicture}}
}

\newcommand{\graph}{
\begin{tikzpicture}
\def\x{0.8};
\def\y{0.4};
\def\s{1.3}

\node (u1)[circle,fill=black,inner sep=\s] at (0,0) {};
\node (v1)[circle,draw=black,inner sep=\s] at (\x*0.5,\y) {};
\node (w1)[circle,fill=black,inner sep=\s,label=above:\footnotesize{\,$a$}] at (\x,0) {}; 
\node (z1)[circle,draw=black,inner sep=\s] at (\x*0.5,-1*\y) {};

\node (u2)[circle,draw=black,,inner sep=\s,label=above:\footnotesize{$b$\,}] at (\x*1.7,0) {};
\node (v2)[circle,fill=black,inner sep=\s] at (\x*2.5,\y) {};
\node (w2)[circle,fill=black,inner sep=\s] at (\x*2.5,-1*\y) {}; 

\draw (u1) 	to [out=60, in=180+30] (v1) 
		   	to [out=-30, in=180-60] (w1) 
		   	to [out=-180+60, in=30] (z1) 
		   	to [out=180-30, in=-60] (u1) ;
\draw (u1) -- (w1) ;
\draw (v1) -- (z1) ;

\draw (u2) to [out=50, in=190] (v2) to [out=-65, in=65](w2)  to [out=170, in=-50] (u2);

\draw (w1) -- (u2);
\end{tikzpicture}
}

\newcommand{\graphMinusZeroBridge}{
\begin{tikzpicture}
\def\x{0.8};
\def\y{0.4};
\def\s{1.3}

\node (u1)[circle,fill=black,inner sep=\s] at (0,0) {};
\node (v1)[circle,draw=black,inner sep=\s] at (\x*0.5,\y) {};
\node (w1)[circle,fill=black,inner sep=\s,label=above:\footnotesize{\,$a$}] at (\x,0) {}; 
\node (z1)[circle,draw=black,inner sep=\s] at (\x*0.5,-1*\y) {};

\node (u2)[circle,draw=black,,inner sep=\s,label=above:\footnotesize{$b$\,}] at (\x*1.7,0) {};
\node (v2)[circle,fill=black,inner sep=\s] at (\x*2.5,\y) {};
\node (w2)[circle,fill=black,inner sep=\s] at (\x*2.5,-1*\y) {}; 

\draw (u1) 	to [out=60, in=180+30] (v1) 
		   	to [out=-30, in=180-60] (w1) 
		   	to [out=-180+60, in=30] (z1) 
		   	to [out=180-30, in=-60] (u1) ;
\draw (u1) -- (w1) ;
\draw (v1) -- (z1) ;

\draw (u2) to [out=50, in=190] (v2) to [out=-65, in=65](w2)  to [out=170, in=-50] (u2);
\end{tikzpicture}
}

\newcommand{\graphMinusOneBridge}{
\begin{tikzpicture}
\def\x{0.8};
\def\y{0.4};
\def\s{1.3}

\node (u1)[circle,fill=black,inner sep=\s] at (0,0) {};
\node (v1)[circle,draw=black,inner sep=\s] at (\x*0.5,\y) {};
\node (w1)[circle,draw=black,inner sep=\s,label=above:\footnotesize{\,$a$}] at (\x,0) {}; 
\node (z1)[circle,draw=black,inner sep=\s] at (\x*0.5,-1*\y) {};

\node (u2)[circle,fill=black,,inner sep=\s,label=above:\footnotesize{$b$\,}] at (\x*1.7,0) {};
\node (v2)[circle,fill=black,inner sep=\s] at (\x*2.5,\y) {};
\node (w2)[circle,fill=black,inner sep=\s] at (\x*2.5,-1*\y) {}; 

\draw (u1) 	to [out=60, in=180+30] (v1) 
		   	to [out=-30, in=180-60] (w1) 
		   	to [out=-180+60, in=30] (z1) 
		   	to [out=180-30, in=-60] (u1) ;
\draw (u1) -- (w1) ;
\draw (v1) -- (z1) ;

\draw (u2) to [out=50, in=190] (v2) to [out=-65, in=65](w2)  to [out=170, in=-50] (u2);

\end{tikzpicture}
}

\section{Introduction}

Resolution is one of the most studied propositional proof systems in proof complexity due to its naturality and it connections to practical SAT solving~\cite{Nordstrom15,BussN2021}. A refutation of a CNF-formula in this system (a resolution refutation) relies uniquely on clausal resolution: in a refutation, clauses are iteratively derived by resolutions on clauses from the formula or previously inferred clauses, until reaching the empty clause indicating unsatisfiability. In this paper, we consider regular resolution which is the restriction of resolution to proofs in which, intuitively, variables which have been resolved away from a clause cannot be reintroduced later on by additional resolution steps. This fragment of resolution is known to generally require exponentially longer refutations than general resolution~\cite{Goerdt93,AlekhnovichJPU07,Urquhart11,VinyalsEJN20} but is still interesting since it corresponds to DPLL-style algorithms~\cite{DavisLL62,DavisP60}. Consequently, there is quite some work on regular resolution, see e.g.~\cite{AtseriasBRLNR18,Urquhart87,BeckI13,BeameBI12} for a very small sample.

Tseitin-formulas are encodings of certain systems of linear equations whose structure is given by a graph~\cite{Tseitin}. They have been studied extensively in proof complexity essentially since the creation of the field because they are hard instances in many settings, see e.g.~\cite{Urquhart87, Ben-Sasson02, ItsyksonO13, ItsyksonRSS19,BeameBI12}. It is known that different properties of the underlying graph characterize different parameters of their resolution refutations~\cite{GalesiTT20, AlekhnovichR11, ItsyksonO13}. Extending this line of work, we here show that treewidth determines the length of regular resolution refutations of Tseitin-formulas: classes of Tseitin-formulas of bounded degree have polynomial length regular resolution refutations if and only if the treewidth of the underlying graphs is bounded logarithmically in their size.
The upper bound for this result was already known from~\cite{AlekhnovichR11} where it is shown that, for every graph~$G$, unsatisfiable Tseitin-formulas with the underlying graph $G$ have regular resolution refutations of length at most $2^{O(\tw(G))}|V(G)|^c$ where $c$ is a constant. We provide a matching lower bound:
\begin{theorem}\label{theorem:main_result}
Let $T(G,c)$ be an unsatisfiable Tseitin-formula where $G$ is a connected graph with maximum degree at most $\Delta$. The length of the smallest regular resolution refutation of $T(G,c)$ is at least $2^{\Omega(\tw(G)/\Delta)}|V(G)|^{-1}$.
\end{theorem}

There were already known lower bounds for the length of resolution refutations of Tseitin-formulas based on treewidth before. For \emph{general} resolution, a $2^{\Omega(\tw(G)^2)/|V(G)|}$ lower bound can be inferred width the classical width-length relation of~\cite{Ben-Sasson02} and width bounds of~\cite{GalesiTT20}. This gives a tight $2^{\Omega(\tw(G))}$ bound when the treewidth of $G$ is linear in its number of vertices. For smaller treewidth, better bounds of $2^{\Omega(\tw(G))/\log|V(G)|}$ that almost match the upper bound where shown in~\cite{ItsyksonRSS19} for \emph{regular} resolution refutations. Building on~\cite{ItsyksonRSS19}, we eliminate the division by $\log|V(G)|$ in the exponent and thus give a tight $2^{\Theta(\tw(G))}$ dependence.

As in~\cite{ItsyksonRSS19}, our proof strategy follows two steps. First, we show that the problem of bounding the length of regular resolution refutations of an \emph{unsatisfiable} Tseitin-formula can be reduced to lower bounding the size of certain representations of a \emph{satisfiable} Tseitin-formula. Itsykson et al.~in \cite{ItsyksonRSS19} used a similar reduction of lower bounds for regular resolution refutations to bounds on read-once branching programs (1-BP) for satisfiable Tseitin-formulas, using the classical connection between regular resolution and the search problem which, given an unsatisfiable CNF-formula and a truth assignment, returns a clause of the formula it falsifies~\cite{LovaszNNW95}. Itsykson et al.~showed that there is a transformation of a 1-BP solving the search problem for an unsatisfiable Tseitin-formula into~a 1-BP of pseudopolynomial size computing a satisfiable Tseitin-formula with the same underlying graph. This yields lower bounds for regular resolution from lower bounds for 1-BP computing satisfiable Tseitin-formulas which~\cite{ItsyksonRSS19} also shows. Our crucial insight here is that when more succinct representations are used to present the satisfiable formula, the transformation from the unsatisfiable instance can be changed to have only a polynomial instead of pseudopolynomial size increase. Concretely, the representations we use are so-called decomposable negation normal forms (DNNF) which are very prominent in the field of knowledge compilation~\cite{Darwiche01} and generalize 1-BP. We show that every refutation of an unsatisfiable Tseitin-formula can be transformed into a DNNF-representation of a satisfiable Tseitin-formula with the same underlying graph with only polynomial overhead.

In a second step, we then show for every satisfiable Tseitin-formula with an underlying graph $G$ a lower bound of $2^{\Omega(\tw(G))}$ on the size of DNNF computing the formula. To this end, we adapt techniques developed in~\cite{BovaCMS16} to a parameterized setting. \cite{BovaCMS16} uses rectangle covers of a function, a common tool from communication complexity, to lower bound the size of any DNNF computing the function. Our refinement takes the form of a two-player game in which the first player tries to cover the models of a function with few rectangles while the second player hinders this construction by adversarially choosing the variable partitions respected by the rectangles from a certain set of partitions. We show that this game gives lower bounds for DNNF, and consequently the aim is to show that the adversarial player can always force $2^{\Omega(\tw(G))}$ rectangles in the game when playing on~a Tseitin-formula with graph $G$. This is done by proving that any rectangle for a carefully chosen variable partition \emph{splits} parity constraints of the formula in a way that bounds by~a function of $\tw(G)$ the number of models that can be covered. We show that, depending on the treewidth of $G$, the adversarial player can choose a partition to limit the number of models of every rectangle constructed in the game to the point that at least $2^{\Omega(\tw(G))}$ of them will be needed to cover all models of the Tseitin-formula. As a consequence, we get the desired lower bound of $2^{\Omega(\tw(G))}|V(G)|^{-1}$ for regular resolution refutations of Tseitin-formulas.

\section{Preliminaries}

\paragraph*{Notions on Graphs.}

We assume the reader is familiar with the fundamentals of graph theory. For a graph~$G$, we denote by $V(G)$ its vertices and by $E(G)$ its edges. For $v \in V(G)$, $E(v)$ denotes the edges incident to~$v$ and $N(v)$ its neighbors ($v$ is not in $N(v)$). For a subset $V'$ of $V(G)$ we denote by $G[V']$ the sub-graph of~$G$ induced by~$V'$.

A binary tree whose leaves are in bijection with the edges of~$G$ is called  a \emph{branch decomposition}\footnote{We remark that often branch decompositions are defined as unrooted trees. However, it is easy to see that our definition is equivalent, so we use it here since it is more convenient in our setting.}. Each edge $e$ of a branch decomposition~$T$ induces a partition of $E(G)$ into two parts as the edge sets that appear in the two connected components of~$T$ after deletion of~$e$. The number of vertices of $G$ that are incident to edges in both parts of this partition is the order of~$e$, denoted by $order(e,T)$. The \emph{branchwidth} of $G$, denoted by $\bw(G)$, is defined as $\bw(G) = \min_T \max_{e \in E(T)} order(e,T)$, where $\min_T$ is over all branch decompositions of~$G$.

While it is convenient to work with branchwidth in our proofs, we state our main result with the more well-known \emph{treewidth} $\tw(G)$ of a graph~$G$. This is justified by the following well-known connection between the two measures.

\begin{lemma}\label{lemma:bw_vs_tw} \textup{\cite[Lemma 12]{HarveyW17}}
If $\bw(G) \geq 2$, then
$\bw(G) - 1 \leq \tw(G) \leq \frac{3}{2}\bw(G)$.
\end{lemma}

A separator $S$ in a connected graph $G$ is defined to be a vertex set such that $G\setminus S$ is non-empty and not connected. A graph $G$ is called $3$-connected if and only if it has at least $4$ vertices and, for every $S \subseteq V(G)$, $|S| \leq 2$, the graph $G \setminus S$ is connected. 

\paragraph*{Variables, assignments, v-trees.}
Boolean variables can have value 0 ($false$) or~1 ($true$). The notation $\ell_x$ refers to a literal for a variable $x$, that is, $x$ or its negation~$\overline{x}$. Given a set $X$ of Boolean variables, $\lit(X)$ denotes its set of literals.~A truth assignment to $X$ is a mapping $a: X \rightarrow \{0,1\}$. 
 If $a_X$ and $a_Y$ are assignments to \emph{disjoint} sets of variables $X$ and $Y$, then $a_X \cup a_Y$ denotes the combined assignment to $X \cup Y$. The set of assignments to $X$ is denoted by $\{0,1\}^X$. Let $f$ be a Boolean function, we denote by $\var(f)$ its variables and by $\sat(f)$ its set of models, i.e., assignments to $\var(f)$ on which~$f$ evaluates to $1$. A v-tree of $X$ is a binary tree $T$ whose leaves are labeled bijectively with the variables in $X$. A v-tree $T$ of $X$ induces a set of partitions $(X_1, X_2)$ of $X$ as follows: choose a vertex $v$ of $T$, setting $X_1$ to contain exactly the variables in $T$ that appear below~$v$ and $X_2:= X\setminus X_1$.

\paragraph*{Tseitin-Formulas.}\label{section:tseitin}

Tseitin formulas are systems of parity constraints whose structure is determined by a graph. Let $G=(V,E)$ be a graph and let $c : V \rightarrow \{0,1\}$ be a labeling of its vertices called a \emph{charge function}. The Tseitin-formula $T(G,c)$ has for each edge $e \in E$ a Boolean variable $x_e$ and for each vertex $v\in V$ a constraint $\chi_v : \sum_{e \in E(v)} x_e = c(v) \mod 2$. The Tseitin-formula $T(G,c)$ is then defined as $T(G,c) := \bigwedge_{v \in V} \chi_v$, i.e., the conjunction of the parity constraints for all $v\in V$. By $\overline{\chi_v}$ we denote the negation of $\chi_v$, i.e., the parity constraint on $(x_e)_{e\in E(v)}$ with charge $1-c(v)$.

\begin{proposition}\textup{~\cite[Lemma 4.1]{Urquhart87}}\label{proposition:satisfiability_tseitin_formula}
The Tseitin-formula $T(G,c)$ is satisfiable if and only if for every connected component~$U$ of~$G$ we have $\sum_{v \in U} c(v) = 0 \mod 2$.
\end{proposition}

\begin{proposition}\textup{\cite[Lemma 2]{GlinskihI17}}\label{proposition:number_model_tseitin_formula}
Let $G$ be a graph with $K$ connected components. If the Tseitin-formula $T(G,c)$ is satisfiable, then it has $2^{|E(G)| - |V(G)| + K}$ models.
\end{proposition}

When conditioning the formula $T(G,c)$ on a literal $\ell_e \in \{x_e,\overline{x_e}\}$ for $e = ab$ in $E(G)$, the resulting function is another Tseitin formula $T(G,c) \vert \ell_e = T(G',c')$ where $G'$ is the graph $G$ without the edge $e$ (so $G' = G - e$) and $c'$ depends on $\ell_e$. If $\ell_e = \overline{x_e}$ then $c'$ equals $c$. If $\ell_e = x_e$ then $c' = c + 1_a + 1_b \mod 2$, where $1_v$ denotes the charge function that assigns $1$ to $v$ and $0$ to all other variables.

Since we consider Tseitin-formulas in the setting of proof systems for CNF-formulas, we will assume in the following that they are encoded as CNF-formulas. In this encoding, every individual parity constraint $\chi_v$ is expressed as a CNF-formula $F_v$ and $T(G,c):= \bigwedge_{v\in V} F_v$. Since it takes $2^{|E(v)|-1}$ clauses to write the parity constraint $\chi_v$, each clause containing $E(v)$ literals, we make the standard assumption that $E(v)$ is bounded, i.e., there is a constant upper bound $\Delta$ on the degree of all vertices in $G$.

\paragraph*{DNNF.} A circuit over $X$ in \emph{negation normal form} (NNF) is a directed acyclic graph whose leaves are labeled with literals in $\lit(X)$ or 0/1-constants, and whose internal nodes are labeled by $\lor$-gates or $\land$-gates. We use the usual semantics for the function computed by (gates of) Boolean circuits.
Every NNF can be turned into an equivalent NNF whose nodes have at most two successors in polynomial time. So we assume that NNF in this paper have only binary gates and thus define the size $|D|$ as the number of gates, which is then at most half the number of wires. Given a gate $g$, we denote by $\var(g)$ the variables for the literals appearing under $g$. When $g$ is a literal input $\ell_x$, we have $\var(g) = \{x\}$, and when it is a 0/1-input, we define $\var(g) = \emptyset$. A gate with two children~$g_l$ and~$g_r$ is called \emph{decomposable} when $\var(g_l) \cap \var(g_r) = \emptyset$, and it is called \emph{complete} (or \emph{smooth}) when $\var(g_l) = \var(g_r)$. An NNF whose $\land$-gates are all decomposable is called a \emph{decomposable NNF (DNNF)}. We call~a DNNF \emph{complete} when all its $\lor$-gates are complete. Every DNNF can be made complete in polynomial time. For every Boolean function~$f$ on finitely many variables, there exists~a DNNF computing~$f$.
 
When representing Tseitin-formulas by DNNF, we will use the following:
\begin{lemma}\label{lemma:reduction_to_T(G,0)}
Let $G$ be a graph and let $c$ and $c'$ be two charge functions such that $T(G,c)$ and $T(G,c')$ are satisfiable Tseitin-formulas. Then $T(G,c)$ can be computed by a DNNF of size $s$ if and only if this is true for $T(G,c')$.
\end{lemma}
\begin{proof}[sketch]
 $T(G,c)$ can be transformed into $T(G,c')$ by substituting some variables by their negations, see~\cite[Proposition~26]{ItsyksonRSS19}. So every DNNF for $T(G,c)$ can be transformed into one for $T(G, c')$ by making the same substitutions. \qed
\end{proof}

Proof trees of a DNNF $D$ are tree-like sub-circuits of $D$ constructed iteratively as follows: we start from the root gate and add it to the proof tree. Whenever~an $\land$-gate is met, both its child gates are added to the proof tree. Whenever a $\lor$-gate is met, exactly one child is is added to the proof tree. Each proof tree of~$D$ computes a conjunction of literals. By distributivity, the disjunction of the conjunctions computed by all proof trees of $D$ computes the same function as~$D$. When $D$ is complete, every variable appears exactly once per proof tree, so every proof tree of a complete DNNF encodes a single model. 

\paragraph*{Branching programs.} A branching program (BP) $B$ is a directed acyclic graph with a single source, sinks that uniquely correspond to the values of a finite set~$Y$, and whose inner nodes, called \emph{decision nodes} 
are each labeled by a Boolean variable $x \in X$ and have exactly two output wires called the 0- and 1-wire pointing to two nodes respectively called its 0- and the 1-child. The variable $x$ appears on a path in $B$ if there is a decision node~$v$ labeled by $x$ on that path. A truth assignment $a$ to $X$ induces a path in $B$ which starts at the source and, when encountering a decision node for a variable $x$, follows the 0-wire (resp. the 1-wire) if $a(x)=0$ (resp.~$a(x)=1$). The BP $B$ is defined to compute the value $y\in Y$ on an assignment $a$ if and only if the path of $a$ leads to the sink labeled with $y$. We denote this value $y$ as $B(a)$. Let $f : X \rightarrow Y$ be a function where~$X$ is a finite set of Boolean variables and $Y$ any finite set. Then we say that $B$ computes $f$ if for every assignment $a\in \{0,1\}^X$ we have $B(a) = f(a)$. We say that a node $v$ in $B$ computes a function $g$ if the BP we get from~$B$ by deleting all nodes that are not reachable from $v$ computes~$g$.
	

Let $R \subseteq \{0,1\}^X \times Y$ be a relation where $Y$ is again finite. Then we say that~a BP $B$ computes $R$ if for every assignment $a$ we have that $(a, B(a))\in R$. Let $T(G,c)$ be an unsatisfiable Tseitin-formula for a graph $G=(V,E)$. Then we define the two following relations: $\search_{T(G,c)}$ consists of the pairs $(a, C)$ such that $a$ is an assignment to $T(G,c)$ that does not satisfy the clause $C$ of $T(G,c)$. The relation $\searchvx(G,c)$ consists of the pairs $(a,v)$ such that $a$ does not satisfy the parity constraint $\chi_v$ of a vertex $v\in V$. Note that $\search_{T(G,c)}$ and $\searchvx(G,c)$ both give a reason why an assignment $a$ does not satisfy $T(G,c)$ but the latter is more coarse: $\searchvx(G,c)$ only gives a constraint that is violated while $\search_{T(G,c)}$ gives an exact clause that is not satisfied.

\paragraph*{Regular Resolution.}

We only introduce some minimal notions of proof complexity here; for more details and references the reader is referred to the recent survey~\cite{BussN2021}.
Let $C_1 = x\lor D_1$ and $C_2 = \overline{x}\lor D_2$ be two clauses such that $D_1, D_2$ contain neither $x$ nor $\overline{x}$. Then the clause $D_1\lor D_2$ is inferred by resolution of $C_1$ and $C_2$ on $x$. A resolution refutation of length $s$ of a CNF-formula $F$ is defined to be a sequence $C_1, \ldots, C_s$ such that $C_s$ is the empty clause and for every $i\in [s]$ we have that $C_i$ is a clause of $F$ or it is inferred by resolution of two clauses $C_j, C_\ell$ such that $j,\ell< i$. It is well-known that $F$ has a resolution refutation if and only if $F$ is unsatisfiable.

To every resolution refutation $C_1, \ldots, C_s$ we assign a directed acyclic graph~$G$ as follows: the vertices of $G$ are the clauses $\{C_i\mid i\in [s]\}$. Moreover, there is an edge $C_jC_i$ in $G$ if and only if $C_i$ is inferred by resolution of $C_j$ and some other clause $C_\ell$ on a variable $x$ in the refutation. We also label the edge $C_jC_i$ with the variable $x$. Note that there might be two pairs of clauses $C_j, C_\ell$ and $C_{j'},C_{\ell'}$ such that resolution on both pairs leads to the same clause $C_i$. If this is the case, we simply choose one of them to make sure that all vertices in $G$ have indegree at most $2$. A resolution refutation is called \emph{regular} if on every directed path in~$G$ every variable $x$ appears at most once as a label of an edge. It is known that there is a resolution refutation of $F$ if and only if a regular resolution refutation of $F$ exists~\cite{DavisP60}, but the latter are in general longer~\cite{AlekhnovichJPU07,Urquhart11}.

In this paper, we will not directly deal with regular resolution proofs thanks to the following well-known result.

\begin{theorem}\textup{\cite{LovaszNNW95}}
 For every unsatisfiable CNF-formula $F$, the length of the shortest regular resolution refutation of $F$ is the size of the smallest $1$-BP computing $\search_F$.
\end{theorem}
Since in our setting, from an unsatisfied clause we can directly inferred an unsatisfied parity constraint, we can use the following simple consequence.
\begin{corollary}\label{corollary:1BP_size_searchvx}
For every unsatisfiable Tseitin-formula $T(G,c)$, the length of the shortest regular resolution refutation of $T(G,c)$ is at least the size of the smallest $1$-BP computing $\searchvx(G,c)$.
\end{corollary}

\section{Reduction From Unsatisfiable to Satisfiable Formulas}

To show our main result, we give a reduction from unsatisfiable to satisfiable Tseitin-formulas as in~\cite{ItsyksonRSS19}. There it was shown that, given a $1$-BP $B$ computing $\searchvx(G,c)$ for an unsatisfiable Tseitin-formula $T(G,c)$, one can construct a $1$-BP $B'$ computing the function of a \emph{satisfiable} Tseitin-formula $T(G, c^*)$ such that $|B'|$ is quasipolynomial in $|B|$. Then good lower bounds on the size of $B'$ yield lower bounds for regular refutation by Corollary~\ref{corollary:1BP_size_searchvx}. To give tighter results, we give a version of the reduction from unsatisfiable to satisfiable Tseitin-formulas where the target representation for $T(G, c^*)$ is not $1$-BP but the more succinct DNNF. This lets us decrease the size of the representation from pseudopolynomial to polynomial which, with tight lower bounds in the later parts of the paper, will yield Theorem~\ref{theorem:main_result}.

\begin{theorem}\label{theorem:from_refutation_to_sat_DNNF}
Let $T(G,c)$ be an unsatisfiable Tseitin-formula where $G$ is connected and let $S$ be the length of its smallest resolution refutation. Then there exists for every satisfiable Tseitin-formula $T(G,c^*)$ a DNNF of size $O(S\times|V(G)|)$ computing it. 
\end{theorem}

In the proof of Theorem~\ref{theorem:from_refutation_to_sat_DNNF}, we heavily rely on results from~\cite{ItsyksonRSS19} in particular the notion of well-structuredness that we present in Section~\ref{sct:well-structured}. In Section~\ref{sct:unsattosat} we will then prove Theorem~\ref{theorem:from_refutation_to_sat_DNNF}.

\subsection{Well-structured branching programs for $\searchvx(G,c)$}\label{sct:well-structured}

In a well-structured 1-BP computing $\searchvx(G,c)$, every decision node~$u_k$ for a variable $x_e$ will compute $\searchvx(G_k,c_k)$ where $G_k$ is a \emph{connected} sub-graph of~$G$ containing the edge $e := ab$, and $c_k$ is a charge function such that $T(G_k, c_k)$ is unsatisfiable. Since $u_k$ deals with $T(G_k,c_k)$, its 0- and 1-successors~$u_{k_0}$ and~$u_{k_1}$ will work on $T(G_k,c_k)|\ell_e$ for $\ell_e = \overline{x_e}$ and $\ell_e = x_e$, respectively. $T(G_k,c_k)|\ell_e$ is a Tseitin-formula whose underlying graph is $G_k - e$ and whose charge function is $c_k$ or $c_k + 1_a + 1_b \mod 2$ depending on $\ell_e$. For convenience, we introduce the notation $\gamma_k(x_e) = c_k + 1_a + 1_b  \mod 2$ and $\gamma_k(\overline{x_e}) = c_k$. Since~$G_k$ is connected, $G_k - e$ has at most two connected components. Let $G^a_k$ and~$G^b_k$ denote the components of $G_k - e$ containing $a$ and $b$, respectively. Note that $G^a_k = G^b_k$ when~$e$ is not a bridge of $G_k$. Let $\gamma^a_k(\ell_e)$ and $\gamma^b_k(\ell_e)$ denote the restriction of $\gamma_k(\ell_e)$ to the vertices of $G^a_k$ and $G^b_k$, respectively. While the graph for $T(G_k,c_k)|\ell_e$ has at most two connected components, exactly one of them holds an odd total charge, so only the Tseitin-formula corresponding to that component is unsatisfiable. Well-structuredness states that $u_{k_0}$ and $u_{k_1}$ each deal with that unique connected component. 

\begin{figure}[t]
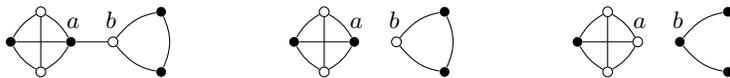

\centering
\begin{subfigure}{.3\textwidth}
 \centering \graph
 \end{subfigure}
\begin{subfigure}{.3\textwidth}
 \centering \graphMinusZeroBridge
 \end{subfigure}
\begin{subfigure}{.3\textwidth}
 \centering \graphMinusOneBridge
 \end{subfigure}
 \caption{The graphs of Example~\ref{ex:wellstructured}. On the left the graph $G_k$, in the middle the result after assigning $0$ to $x_e$, on the right after assigning $1$ to $x_e$.}\label{fig:example}
\end{figure}

\begin{example}\label{ex:wellstructured}
Consider the graph $G_k$ shown on the left in Figure~\ref{fig:example}. Black nodes have charge~$0$ and white nodes have charge~$1$. The corresponding Tseitin-formula $T(G_k,c_k)$ is unsatisfiable because there is an odd number of white nodes. Let $e := ab$. Then $T(G_k,c_k)|\overline{x_e}$ is the Tseitin-formula for the graph $G_k - e$ with charges as shown in the middle of Figure~\ref{fig:example}. Note that  $T(G_k,c_k)|\overline{x_e}$ is unsatisfiable because of the charges in the triangle component $G_k^b$. The repartition of charges for $T(G_k,c_k)|x_e$ illustrated on the right of Figure~\ref{fig:example} shows that $T(G_k,c_k)|x_e$ is unsatisfiable because of the charges in the rombus component~$G_k^a$. Well-structuredness will ensure that, if $u_k$ computes $\searchvx(G_k,c_k)$ and decides $x_e$, then $u_{k_0}$ computes $\searchvx(G_k^b,\gamma^b_k(\overline{x_e}))$ and $u_{k_1}$ computes $\searchvx(G_k^a,\gamma^a_k(x_e))$.
\end{example}

\begin{definition}
Let $T(G,c)$ be an unsatisfiable Tseitin-formula where $G$ is a connected graph. A branching program $B$ computing $\searchvx(G,c)$ is well-structured when, for all nodes $u_k$ of $B$, there exists a connected subgraph $G_k$ of $G$ and a charge function $c_k$ such that $T(G_k, c_k)$ is unsatisfiable, $u_k$ computes $\searchvx(G_k,c_k)$ and 
\begin{enumerate}
\item if $u_k$ is the source, then $G_k = G$ and $c_k = c$,
\item if $u_k$ is a sink corresponding to $v \in V(G)$, then $G_k = (\{v\},\emptyset)$ and $c_k = 1_v$,
\item if $u_k$ is a decision node for $x_{ab}$ with 0- and 1- successors $u_{k_0}$ and $u_{k_1}$, set $\ell_0 = \overline{x_{ab}}$ and $\ell_1 = x_{ab}$, then for all $i \in \{0,1\}$, $(G_{k_i},c_{k_i}) = (G^a_k,\gamma^a_k(\ell_i))$ if $T(G^a_k,\gamma^a_k(\ell_{i}))$ is unsatisfiable, otherwise $(G_{k_i},c_{k_i}) = (G^b_k,\gamma^b_k(\ell_i))$.
\end{enumerate}
\end{definition}

We remark that our definition is a slight simplification of that given by Itsykson et al.~\cite{ItsyksonRSS19}. It can easily be seen that ours is implied by theirs (see Definition 11 and Proposition 16 in~\cite{ItsyksonRSS19}). 

\begin{lemma}\textup{\cite[Lemma 17]{ItsyksonRSS19}}\label{lemma:1BP_are_well_structured} Let $T(G,c)$ be an unsatisfiable Tseitin-formula where $G$ is connected and let $B$ be a 1-BP of minimal size\footnote{\cite[Lemma 17]{ItsyksonRSS19} is for \emph{locally minimal} 1-BP, which encompass minimal size 1-BP.} computing the relation $\searchvx(G,c)$. Then $B$ is well-structured.
\end{lemma}

\subsection{Constructing DNNF from Well-structured branching programs}\label{sct:unsattosat}

Similarly to Theorem 14 in~\cite{ItsyksonRSS19}, we give a reduction from a well-structured 1-BP for $\searchvx(G,c)$ to a DNNF computing a \emph{satisfiable} formula $T(G,c^*)$.

\begin{lemma}\label{lemma:from_well_struct_1BP_to_DNNF}
Let $G$ be a connected graph. Let $T(G,c^*)$ and $T(G,c)$ be Tseitin-formulas where $T(G,c^*)$ is satisfiable and $T(G,c)$ unsatisfiable. For every well-structured 1-BP $B$ computing $\searchvx(G,c)$ there exists a DNNF of size $O(|B| \times |V(G)|)$ computing $T(G,c^*)$.
\end{lemma}
\begin{proof}
Let $S = |B|$ and denote by $u_1,\dots,u_S$ the nodes of $B$ such that if $u_j$ is a successor of $u_i$, then $j < i$ (thus $u_S$ is the source of $B$). For every $i \in [S]$, the node $u_i$ computes $\searchvx(G_i,c_i)$. We will show how to iteratively construct DNNF $D_1, \dots, D_S$ such that, $D_1 \subseteq D_2 \subseteq \dots \subseteq D_S$ and, for every $i \in [ S ]$,
\begin{center}
\quad\quad  for all $v \in V(G_i)$, there is a gate $g_v$ in $D_i$ computing $T(G_i,c_i + 1_v)$. \quad  $(\ast)$
\end{center}
Observe that, since $T(G_i,c_i)$ is unsatisfiable, $T(G_i,c_i + 1_v)$ is satisfiable for any $v \in V(G_i)$. 
We show by induction on $i$ how to construct  $D_i$ by extending $D_{i-1}$ while respecting $(\ast)$. 

For the base case, $u_1$ is a sink of $B$, so it computes $\searchvx(G_v,1_v)$ where $G_v := (\{v\},\emptyset)$ for a vertex $v \in V(G)$. Thus we define $D_1$ as a single constant-1-node which indeed computes $T(G_v,1_v+1_v) = T(G_v,0)$. So $D_1$ is a DNNF respecting $(\ast)$.

Now for the inductive case, suppose we have the DNNF $D_{k-1}$ satisfying~$(\ast)$. Consider the node $u_k$ of $B$. If $u_k$ is a sink of $B$, then we argue as for $D_1$ but since we already have the constant-1-node in $D_{k-1}$ we define $D_k := D_{k-1}$. 

Now assume that $u_k$ is a decision node for the variable $x_e$ with 0- and 1-successors $u_{k_0}$ and $u_{k_1}$. Recall that $u_k$ computes $\searchvx(G_k,c_k)$ and let $e = ab$. There are two cases. If $e$ is not a bridge in $G_k$ then $G^a_k = G^b_k = G_k - e$ and, by well-structuredness,
\begin{itemize}
\item $u_{k_0}$ computes $\searchvx(G_k - e,c_k)$
\item $u_{k_1}$ computes $\searchvx(G_k - e,c_k + 1_a + 1_b)$
\end{itemize}
For every $v \in V(G_k)$, since $k_0, k_1 < k$, by induction there is a gate $g^0_v$ in $D_{k_0}$ computing $T(G_k - e, c_k + 1_v)$ and a gate $g^1_v$ in $D_{k_1}$ computing $T(G_k - e, c_k + 1_a + 1_b + 1_v)$. So for every $v \in V(G_k)$ we add to $D_{k-1}$ an $\lor$-gate $g_v$ whose left input is $\overline{x_e} \land g^0_v$ and whose right input is $x_e \land g^1_v$. By construction, $g_v$ computes $T(G_k,c_k+1_v)$ and the new $\land$-gates are decomposable since $e$ is not an edge of $G_k - e$ and therefore $x_e$ and $\overline{x_e}$ do not appear in $D_{k_0}$ and $D_{k_1}$.

Now if $e = ab$ is a bridge in $G_k$, by well-structuredness, there exist $i \in \{0,1\}$ and $\ell_e \in \{\overline{x_e},x_e\}$ such that 
\begin{itemize}
\item $u_{k_i}$ computes $\searchvx(G^a_k,\gamma^a_k(\ell_e))$ 
\item $u_{k_{1-i}}$ computes $\searchvx(G^b_k,\gamma^b_k(\overline{\ell_e}))$
\end{itemize}
We construct a gate $g_v$ computing $T(G_k,c_k + 1_v)$ for each $v \in V(G_k)$. Assume, without loss of generality, that $v \in V(G^a_k)$, then 
\begin{itemize}
\item $T(G_k,c_k + 1_v) | \overline{\ell_e} \equiv T(G^a_k, \gamma^a_k(\overline{\ell_e}) + 1_v) \land T(G^b_k,\gamma^b_k(\overline{\ell_e})) \equiv 0$ \\ (because of the second conjunct which is known to be unsatisfiable), and
\item $T(G_k,c_k + 1_v) | \ell_e \equiv T(G^a_k, \gamma^a_k(\ell_e) + 1_v) \land T(G^b_k,\gamma^b_k(\ell_e))$
\end{itemize}
For the second item, since $k_0,k_1 < k$, by induction there is a gate $g^i_v$ in $D_{k_i}$ computing $T(G^a_k, \gamma^a_k(\ell_e) + 1_v)$ and there is a gate $g^{1-i}_b$ in $D_{k_{1-i}}$ computing  $T(G^b_k,\gamma^b_k(\overline{\ell_e}) + 1_b)$. But $\gamma_k(\ell_e)= \gamma_k(\overline{\ell_e}) + 1_a + 1_b \mod 2$, so $\gamma^b_k(\ell_e) =  \gamma^b_k(\overline{\ell_e}) + 1_b \mod 2$, therefore $g^{i-1}_b$ computes the formula $ T(G^b_k,\gamma^b_k(\ell_e))$. So we add an $\land$-gate $g_v$ whose left input is $\ell_e$ and whose right input is $s^i_v \land s^{1-i}_b$ and add it to $D_{k-1}$. Note that $\land$-gates are decomposable since $G^a_k$ and $G^b_k$ share no edge and therefore $D_{k_0}$ and $D_{k_1}$ are on disjoint sets of variables. 

Let $D_k$ be the circuit after all $g_v$ have been added to $D_{k-1}$. It is a DNNF satisfying both $D_{k-1} \subseteq D_k$ and $(\ast)$. 

It only remains to bound $|D_S|$. To this end, observe that when constructing $D_k$ from $D_{k-1}$ we add at most $3\times |V_k|$ gates, so $|D_S|$ is at most $3(|V_1|+ \dots + |V_S|) = O(S \times |V(G)|)$. Finally, take any root of $D_S$ and delete all gates not reached from it, the resulting circuit is a DNNF $D$ computing a satisfiable Tseitin formula $T(G,c')$. We get a DNNF computing $T(G,c^*)$ using Lemma~\ref{lemma:reduction_to_T(G,0)}.\qed
\end{proof}

Combining Corollary~\ref{corollary:1BP_size_searchvx}, Lemma~\ref{lemma:1BP_are_well_structured} and Lemma~\ref{lemma:from_well_struct_1BP_to_DNNF} yields Theorem~\ref{theorem:from_refutation_to_sat_DNNF}.

\section{Adversarial Rectangle Bounds}

In this section, we introduce the game we will use to show DNNF lower bounds for Tseitin formulas.
It is based on combinatorial rectangles, a basic object of study from communication complexity.

\begin{definition} 
A \emph{(combinatorial) rectangle} for a variable partition $(X_1,X_2)$ of a variables set $X$ is defined to be a set of assignments of the form $R = A \times B$ where $A \subseteq \{0,1\}^{X_1}$ and $B\subseteq \{0,1\} ^{X_2}$. The rectangle is called balanced when $\frac{|X|}{3} \leq |X_1|, |X_2| \leq \frac{2|X|}{3}$.
\end{definition}

A rectangle on variables $X$ may be seen as a function whose satisfying assignments are exactly the $a \cup b$ for $a \in A$ and $b \in B$, so we sometimes interpret rectangles as Boolean functions whenever it is convenient.

\begin{definition} 
Let $f$ be a Boolean function. A \emph{balanced rectangle cover} of $f$ is a collection $\mathcal{R} = \{R_1,\dots,R_K\}$ of balanced rectangles on $\var(f)$, possibly for different partitions of $\var(f)$, such that $f$ is equivalent to $\bigvee_{i = 1}^K R_i$.
The minimum number of rectangles in a balanced cover of $f$ is denoted by~$R(f)$.
\end{definition}

\begin{theorem}\textup{\cite{BovaCMS16}}\label{thm:bovaetal}
Let $D$ be a DNNF computing a function $f$, then $R(f) \leq \vert D \vert$.
\end{theorem}

When trying to show parameterized lower bounds with Theorem~\ref{thm:bovaetal}, one often runs into the problem that it is somewhat inflexible: the partitions of the rectangles in covers have to be balanced, but in parameterized applications this is often undesirable. Instead, to show good lower bounds, one wants to be able to partition in places that allow to cut in complicated subparts of the problem. This is e.g.~the underlying technique in~\cite{Razgon16}. To make this part of the lower bound proofs more explicit and the technique more reusable, we here introduce a refinement of~Theorem~\ref{thm:bovaetal}.

We define the adversarial multi-partition rectangle cover game for a function~$f$ on variables $X$ and a set $S\subseteq \sat(f)$ to be played as follows: two players, the cover player Charlotte and her adversary Adam, construct in several rounds~a set $\mathcal{R}$ of combinatorial rectangles that cover the set $S$ respecting $f$ (that is, rectangles in $\mathcal{R}$ contain only models of $f$). The game starts with $\mathcal R$ as the empty set. Charlotte starts a round by choosing an input $a\in S$ and a v-tree $T$ of $X$. Now Adam chooses a partition $(X_1, X_2)$ of $X$ induced by $T$. Charlotte ends the round by adding to $\mathcal R$ a combinatorial rectangle for this partition and respecting~$f$ that covers $a$. The game is over when $S$ is covered by $\mathcal R$. The adversarial multi-partition rectangle complexity of $f$ and $S$, denoted by $aR(f,S)$ is the minimum number of rounds in which Charlotte can finish the game, whatever the choices of Adam are. The following theorem gives the core technique for showing lower bounds later on.

\begin{theorem}\label{thm:DNNFlower}
Let $D$ be a complete DNNF computing a function $f$ and let $S\subseteq \sat(f)$. Then $aR(f,S)\le |D|$. 
\end{theorem}
\begin{proof}
Let $X = \var(D)$. We iteratively delete vertices from $D$ and construct rectangles. The approach is as follows: Charlotte chooses an assignment~$a\in S$ not yet in any rectangle she constructed before and a proof tree $T$ accepting~$a$ in~$D$. By completeness of $D$, all variables of $X$ appear exactly once in $T$. Charlotte constructs a v-tree of $X$ from $T$ by deleting negations on the leaves, contracting away nodes with a single child and forgetting the labels of all operation gates. Now Adam chooses a partition induced by $T$ given by a subtree of $T$ with root~$v$. Note that $v$ is a gate of $C$. Let $\sat(D,v) \subseteq \sat(f)$ be the assignments to~$X$ accepted by a proof tree of $C$ passing through $v$, and observe that $\sat(D,v)$ is  a combinatorial rectangle $A \times B$ with $A \subseteq \{0,1\}^{\var(v)}$ and $B \subseteq \{0,1\}^{X \setminus \var(v)}$. Charlotte chooses the rectangle $\sat(D,v)$, deletes it from $S$ and the game continues.
 
 Note that the vertex $v$ in the above construction is different for every iteration of the game: by construction, Charlotte never chooses an assignment $a$ that is in any set $\sat(D,v)$ for a vertex $v$ that has appeared before. Thus, no such $v$ can appear in the proof tree of the chosen $a$. Consequently, a new vertex $v$ is chosen for each assignment $a$ that Charlotte chooses and thus the game will never last more than $|D|$ rounds. \qed
\end{proof}



\section{Splitting Parity Constraints}

In this section, we will see that rectangles \emph{split} parity constraints in a certain sense and show how this is reflected in in the underlying graph of Tseitin-formulas. This will be crucial in proving the DNNF lower bound in the next section with the adversarial multi-partition rectangle cover game.

\subsection{Rectangles Induce Sub-Constraints for Tseitin-Formulas}

Let $R$ be a rectangle for the partition $(E_1,E_2)$ of $E(G)$ such that $R \subseteq \sat(T(G,c))$. Assume that there is a vertex $v$ of $G$ incident to edges in $E_1$ and to edges in $E_2$, i.e., $E(v) = E_1(v) \cup E_2(v)$ where neither $E_1(v)$ not $E_2(v)$ is empty. We will show that $R$ does not only respect $\chi_v$, but it also respects a sub-constraint of~$\chi_v$.

\begin{definition}
Let $\chi_v$ be a parity constraint on $(x_e)_{e \in E(v)}$. A sub-constraint of~$\chi_v$ is a parity constraint $\chi'_v$ on a non-empty proper subset of the variables of $\chi_v$.
\end{definition}

\begin{lemma}\label{lemma:rectangle_sub_constraints}
Let $T(G,c)$ be a satisfiable Tseitin-formula and let $R$ be a rectangle for the partition $(E_1,E_2)$ of $E(G)$ such that $R \subseteq \sat(T(G,c))$. If $v \in V(G)$ is incident to edges in $E_1$ and to edges in $E_2$, then there exists a sub-constraint $\chi'_v$ of $\chi_v$ such that $R \subseteq \sat(T(G,c) \land \chi'_v)$.
\end{lemma}
\begin{proof}
Let $a_1 \cup a_2 \in R$ where $a_1$ is an assignment to $E_1$ and $a_2$ an assignment to~$E_2$. Let $a_1(v)$ and $a_2(v)$ denote the restriction of $a_1$ and $a_2$ to $E_1(v)$ and $E_2(v)$, respectively. We claim that for all $a'_1 \cup a'_2 \in R$, we have that $a'_1(v)$ and $a_1(v)$ have the same parity, that is, $a_1(v)$ assigns an odd number of variables of $E_1(v)$ to 1 if and only if it is also the case for $a'_1(v)$. Indeed if $a_1(v)$ and $a'_1(v)$ have different parities, then so do $a_1(v) \cup a_2(v)$ and $a'_1(v) \cup a_2(v)$. So either $a_1 \cup a_2$ or $a'_1 \cup a_2$ falsifies $\chi_v$, but both assignments are in $R$, so $a_1(v)$ and $a'_1(v)$ cannot have different parities as this contradicts $R \subseteq \sat(T(G,c))$. Let $c_1$ be the parity of $a_1(v)$, then we have that assignments in $R$ must satisfy $\chi'_v : \sum_{e \in E_1(v)} x_e = c_1 \mod 2$, so $R \subseteq \sat(T(G,c) \land \chi'_v)$. \qed
\end{proof}

Renaming $\chi'_v$ as $\chi^1_v$ and adopting notations from the proof, one sees that $\chi^1_v \land \chi_v \equiv \chi^1_v \land \chi^2_v$ where $\chi^2_v : \sum_{e \in E_2(v)} x_e = c(v) + c_1 \mod 2$. So $R$ respects the formula $(T(G,c) - \chi_v) \land \chi^1_v \land \chi^2_v$ where $(T(G,c) - \chi_v)$ is the formula obtained by removing all clauses of $\chi_v$ from $T(G,c)$. In this sense, the rectangle is splitting the constraint $\chi_v$ into two subconstraints in disjoint variables. Since $\chi_v \equiv (\chi^1_v \land \chi^2_v) \lor (\overline{\chi}^1_v \land \overline{\chi}^2_v)$ it is plausible that potentially many models of $\chi_v$ are not in~$R$. We show that this is true in the next section.

\subsection{Vertex Splitting and Sub-constraints for Tseitin-Formulas}

Let $v \in V(G)$ and let $(N_1, N_2)$ be a proper partition of $N(v)$, that is, neither~$N_1$ nor $N_2$ is empty. The graph $G'$ we get by \emph{splitting} $v$ along $(N_1, N_2)$ is defined as the graph we get by deleting $v$, adding two vertices~$v^1$ and~$v^2$, and connecting~$v^1$ to all vertices in $N_1$ and $v^2$ to all vertices in $N_2$. We now show that splitting a vertex $v$ in a graph $G$ has the same effect as adding a sub-constraint of $\chi_v$.

\begin{lemma}\label{lemma:graph_splitting_equals_subconstraint}
Let $T(G,c)$ be a Tseitin-formula. Let $v \in V(G)$ and let $(N_1, N_2)$ be a proper partition of~$N(v)$. Let $c_1$ and $c_2$ be such that $c_1 + c_2 = c(v) \mod 2$ and let $\chi^i_v : \sum_{u \in N_i} x_{uv} = c_i \mod 2$ for $i \in \{1,2\}$ be sub-constraints of $\chi_v$. Call~$G'$ the result of splitting~$v$ along $(N_1, N_2)$ and set
 \begin{align*}
  c'(u) := \begin{cases}
            c(u), &\text{ if } u \in V(G) \setminus \{v\}\\
            c_i, &\text{ if } u = v^i, i \in \{1, 2\}
           \end{cases}
 \end{align*}
There is a bijection $\rho : \var(T(G,c)) \rightarrow \var(T(G',c'))$ acting as a renaming of the variables such that $T(G',c') \equiv (T(G,c) \land \chi^1_v) \circ \rho$.
\end{lemma}
\begin{proof}
Denote by $T(G,c) - \chi_v$ the formula equivalent to the conjunction of all~$\chi_u$ for $u \in V(G) \setminus \{v\}$. Then $T(G,c) \land \chi^1_v \equiv (T(G,c) - \chi_v) \land \chi^1_v \land \chi^2_v$. The constraints $\chi_u$ for $u \in V(G) \setminus \{v\}$ appear in both $T(G',c')$ and in $T(G,c) - \chi_v$ and the sub-constraints $\chi^1_v$ and $\chi^2_v$ are exactly the constraints for $v^1$ and $v^2$ in $T(G',c')$ modulo the variable renaming $\rho$ defined by $\rho (x_{uv}) = x_{uv^1}$ when $u \in N_1$, $\rho (x_{uv}) = x_{uv^2}$  when $u \in N_2$, and $\rho (x_e) = x_e$ when $v$ is not incident to $e$.
\qed
\end{proof}

Intuitely, Lemma~\ref{lemma:graph_splitting_equals_subconstraint} says that splitting a vertex in $G$ and adding sub-constraint are essentially the same operation. This allows us to compute the number of models of a Tseitin-formula to which a sub-constraint was added.

\begin{lemma}\label{lemma:graph_splitting_to_connected_equals_half_models}
Let $T(G,c)$ be a satisfiable Tseitin-formula where $G$ is connected. Define $T(G',c')$ as in Lemma~\ref{lemma:graph_splitting_equals_subconstraint}. If $G'$ is connected then $T(G',c')$ has $2^{|E(G)| - |V(G)|}$ models. 
\end{lemma}
\begin{proof}
By Proposition~\ref{proposition:satisfiability_tseitin_formula}, $T(G',c')$ is satisfiable since $T(G,c)$ is satisfiable and $\sum_{u \in V(G')} c'(u) = \sum_{u \in V(G)} c(u) = 0 \mod 2$. Using Proposition~\ref{proposition:number_model_tseitin_formula} yields that $T(G',c')$ has $2^{|E(G')| - |V(G')| + 1} = 2^{|E(G)| - |V(G)|}$ models. 
\qed
\end{proof}

\begin{lemma}\label{lemma:graph_splitting_a_lot_to_connected_equals_far_less_models}
Let $T(G,c)$ be a satisfiable Tseitin-formula where $G$ is connected. Let $\{v_1, \ldots, v_k\}$ be an independent set in $G$. For all $i \in [k]$ let $(N_1^i, N_2^i)$ be a proper partition of  $N(v_i)$ and let $\chi'_{v_i} : \sum_{u \in N^i_1} x_{uv_i} = c_i \mod 2$. If the graph obtained by splitting all $v_i$ along $(N_1^i, N_2^i)$ is connected, then the formula $T(G,c) \land \chi'_{v_1} \land \dots \land \chi'_{v_k}$ has $2^{|E(G)| - |V(G)| - k + 1}$  models.
\end{lemma}
\begin{proof}
An easy induction based on Lemma~\ref{lemma:graph_splitting_equals_subconstraint} and Lemma~\ref{lemma:graph_splitting_to_connected_equals_half_models}. The induction works since, $\{v_1, \ldots, v_k\}$ being an independant set, the edges to modify by splitting $v_i$ are still in the graph where $v_1,\dots,v_{i-1}$ have been split.
\qed
\end{proof}

\subsection{Vertex Splitting in 3-Connected Graphs}

When we want to apply the results of the last sections to bound the size of rectangles, we require that the graph $G$ remains connected after splitting vertices. This is obviously not true for all choices of vertex splits, but here we will see that if $G$ is sufficiently connected, then we can always chose a large subset of any set of potential splits such that, after applying the split for this subset, $G$ remains connected. 

\begin{lemma}\label{lemma:choose_vertices_from_3-connected_graph}
Let $G$ be a $3$-connected graph of and let $\{v_1, \ldots, v_k\}$ be an independent set in $G$. For every $i \in [k]$ let $(N_1^i, N_2^i)$ be a proper partition of  $N(v_i)$. Then there is a subset $S$ of $\{v_1, \ldots, v_k\}$ of size at least $k/3$ such that the graph resulting from splitting all $v_i\in S$ along the corresponding $(N_1^i, N_2^i)$ is connected.
\end{lemma}
\begin{proof}
 Let $C_1, \ldots, C_r$ be the connected components of the graph $G_1$ that we get by splitting \emph{all} $v_i$. If $G_1$ is connected, then we can set $S=\{v_1, \ldots, v_k\}$ and we are done. So assume that $r>1$ in the following. Now add for every $i\in [k]$ the edge $(v^1_i,v_i^2)$. Call this edge set $L$ (for \emph{links}) and the resulting graph $G_2$. Note that $G_2$ is connected and for every edge set $E'\subseteq L$ we have that $G_2\setminus E'$ is connected if and only if $G$ is connected after splitting the vertices corresponding to the edges in $E'$. Denote by $L_{in}$ the edges in $L$ whose end points both lie in some component $C_j$ and let $L_{out}:= L\setminus L_{in}$.
 
 We claim that for every $C_j$, at least three edges in $L_{out}$ are incident to a vertex in $C_j$. Since $G_2$ is connected but the set $C_j$ is a connected component of $G_2 \setminus L = G_1$, there must be at least one edge in $L$ incident to a vertex in $C_j$. That vertex is by construction one of $v_1, \ldots, v_k$, say it is $v_i$. Since $N_1^i\ne \emptyset$ and $N_2^i\ne \emptyset$, we have that $v_i$ has a neighbor $w$ in $C_j$ and, $w \not\in \{v_1, \ldots, v_k\}$ since it is an independent set. Now let $L^j_{out}$ be the edges in $L_{out}$ that have an end point in $C_j$. Note that if we delete the vertices $S^j \subseteq \{v_1, \ldots, v_k\}$ for which the edges in $L^j_{out}$ were introduced in the construction of $G_2$, then a subset of $C_j$ becomes disconnected from the rest of the graph (which is non-empty because there is at least one component different from $C_j$ in $G_2$ which also contains a vertex not in $\{v_1, \ldots, v_k\}$ by the same reasoning as before). But then, because $G$ is $3$-connected, there must be at least three edges in $L^j_{out}$. Let $k':= |L_{out}|$, then by the handshaking lemma, 
\begin{align*}
  r \le \frac{2}{3} k'.
\end{align*}
Now contract all components $C_i$ in $G_2$ and call the resulting graph $G_3$. Note that~$G_3$ is connected and that $E(G_3)= L_{out}$. Moreover, whenever $G_3\setminus E^*$ is connected for some $E^*\subseteq L_{out}$, then $G$ is connected after splitting the corresponding vertices. Choose any spanning tree $T$ of $G_3$. Then $|E(T)|=r-1$ and deleting $E^* := L_{out} \setminus E(T)$ leaves $G_3$ connected. Thus the graph $G^*$ we get from~$G$ after splitting the vertices corresponding to $E^*$ is connected. We have 
 \begin{align*}
  |E^*| = |L_{out}| - |E(T)| = k'- (r-1) > \frac{k'}{3}.
 \end{align*}
Now observe that in $G$ we can safely split all $k-k'$ vertices $v_i$ that correspond to edges $v_i^1v_i^2$ such that $v_i^1$ and $v_i^2$ lie in the same component of $G_1$ without disconnecting the graph. Thus, overall we can split a set of size 
\begin{align*}
 k-k' + |E^*| > k-k'+ \frac{k'}{3} \ge \frac{k}{3}
\end{align*}
in $G$ such that the resulting graph remains connected. \qed
\end{proof}

\section{DNNF Lower Bounds for Tseitin-Formulas}

In this section, we use the results of the previous sections to show our lower bounds for DNNF computing Tseitin-formulas. To this end, we first show that we can restrict ourselves to the case of $3$-connected graphs.

\subsection{Reduction from Connected to 3-Connected Graphs}

In~\cite{BodlaenderK06}, Bodlaender and Koster study how separators can be used in the context of treewidth. They call a separator $S$ \emph{safe for treewidth} if there exists a connected component of $G \setminus S$ whose vertex set $V'$ is such that $\tw(G[S \cup V'] + clique(S)) = \tw(G)$, where $G[S \cup V'] + clique(S)$ is the graph induced on $S \cup V'$ with additional edges that pairwise connect all vertices in $S$.

\begin{lemma}\textup{~\cite[Corollary 15]{BodlaenderK06}}\label{lemma:safe_size_1_and_size_2_separators} Every separator of size 1 is safe for treewdith. When $G$ has no separator of size 1, every separator of size 2 is safe for treewidth.
\end{lemma}

Remember that a \emph{topological minor} $H$ of a $G$ is a graph that can be constructed from $G$ by iteratively applying the following operations:
\begin{enumerate}[leftmargin=*]
 \item[$-$] edge deletion,
 \item[$-$] deletion of isolated vertices, or
 \item[$-$] subdivision elimination: if $\deg(v) = 2$ delete $v$ and connect its two neighbors.
\end{enumerate}

\begin{lemma}\label{lemma:DNNF_size_for_TS_on_topological_minors} 
Let $H$ be a topological minor of $G$. If the satisfiable Tseitin-formula $T(G, 0)$ has a DNNF of size $s$, then so does $T(H, 0)$.
\end{lemma}
\begin{proof}
Edge deletion corresponds to conditioning the variable by $0$ so it cannot increase the size of a DNNF. Deletion of an isolated vertex  does not change the Tseitin-formula. Finally, let $e_1, e_2$ be the edges incident to a vertex of degree $2$. Since we assume that all charges $c(v)$ are $0$, in every satisfying assignment, $x_{e_1}$ and $x_{e_2}$ take the same value. Thus we can simply forget the variable of $x_{e_2}$ which does not increase the size of a DNNF~\cite{DarwicheM02}.
\qed
\end{proof}

\begin{lemma}\label{lemma:tological_minor_3-connected}
Let $G$ be a graph with treewidth at least $3$. Then $G$ has a $3$-connected topological minor $H$ with $\tw(H) = \tw(G)$. 
\end{lemma}
\begin{proof}
First we construct a topological minor of~$G$ with no separator of size $1$ that preserves treewidth. Let $S = \{v\}$ be a separator of size 1 of $G$, then $G \setminus S$ has a connected component $V'$ such that $G[S \cup V'] + clique(S) = G[S \cup V']$ has treewidth $\tw(G)$. Let $G' = G[S \cup V']$, then $\tw(G') = \tw(G)$. Observe that~$G'$ is a topological minor (remove all edges not in $G[S \cup V']$ thus isolating all vertices not in $S \cup V'$, which are then deleted) where $S$ is no longer a separator. Repeat the construction until~$G'$ has no separator of size 1.

Now assume $S = \{u,v\}$ is a separator of $G'$. If $V'$ are the vertices of a connected component of $G' \setminus S$, then there is a path from $u$ to $v$ in $G[S \cup V']$ since otherwise either $\{u\}$ or $\{v\}$ is a separator of size $1$ of~$G'$. Lemma~\ref{lemma:safe_size_1_and_size_2_separators} ensures that there is a connected component~$H'$ in~$G' \setminus S$ such that $H := (V(H') \cup S, E(H') \cup \{uv\})$ has treewidth $\tw(H) = \tw(G') = \tw(G)$. Let us prove that $H$ is topological minor of~$G'$. Consider a connected component of $G' \setminus S$ distinct from $H'$ with vertices $V'$ and let $P$ be a path connecting~$u$ to~$v$ in $G[S \cup V']$. Delete all edges of $G[S \cup V']$ not in~$P$, then delete all isolated vertices in $V'$ so that only~$P$ remains, finally use subdivision elimination to reduce $P$ to a single edge $uv$. Repeat the procedure for all connected components of $G' \setminus S$ distinct from $H'$, the resulting topological minor is $G[V(H') \cup S]$ with the (additional) edge $uv$, so $H$. 

Repeat the construction until there are no separators of size $1$ or size $2$ left. Note that this process eventually terminates since the number of vertices decreases after every separator elimination. The resulting graph $H$ is a topological minor of $G$ of treewidth $\tw(G)$ without separators of size $1$ or $2$. Since $\tw(H) = \tw(G) \ge 3$, we have that $H$ has at least $4$ vertices, so $H$ is $3$-connected.
\qed
\end{proof}

\subsection{Proof of the DNNF Lower Bound and of the Main Result}

\begin{lemma}\label{lem:dnnf_lower}
Let $T(G,c)$ be a satisfiable Tseitin-formula where $G$ is a connected graph with maximum degree at most $\Delta$. Any complete DNNF computing $T(G,c)$ has size at least~$2^{\Omega(\tw(G)/\Delta)}$.
\end{lemma}
\begin{proof}
By Lemma~\ref{lemma:reduction_to_T(G,0)} we can set $c = 0$. By Lemmas~\ref{lemma:DNNF_size_for_TS_on_topological_minors} and ~\ref{lemma:tological_minor_3-connected} we can assume that $G$ is 3-connected. We show that the adversarial multi-partition rectangle complexity is lower-bounded by $2^{k}$ for $k:= \frac{2\tw(G)}{9\Delta}$. To this end, we will show that the rectangles that Charlotte can construct after Adam's answer are never bigger than $2^{|E(G)|-|V(G)|-k+1}$. Since $T(G,c)$ has $2^{|E(G)|-|V(G)|+1}$ models, the claim then follows.

So let Charlotte choose an assignment $a$ and a v-tree $T$. Note that since the variables of $T(G,0)$ are the edges of $G$, the v-tree $T$ is also a branch decomposition of $G$. Now by the definition of branchwidth, Adam can choose a cut of $T$ inducing a partition $(E_1, E_2)$ of $E(G)$ for which there exists a set $V' \in V(G)$ of at least $\bw(G) \geq \frac{2}{3}\tw(G)$ vertices incident to edges in $E_1$ and to edges in $E_2$. 


$G$ has maximum degree $\Delta$ so there is an independent set $V'' \subset V'$ of size at least $\frac{|V'|}{\Delta}$. Since $G$ is $3$-connected, by Lemma~\ref{lemma:choose_vertices_from_3-connected_graph} there is a subset $V^* \subseteq V''$ of size at least $\frac{|V''|}{3} \geq \frac{2\tw(G)}{9\Delta} = k$ such that~$G$ remains connected after splitting of the nodes in $V^*$ along the partition of their neighbors induced by the edges partition $(E_1,E_2)$. Using Lemma~\ref{lemma:rectangle_sub_constraints}, we find that any rectangle $R$ for the partition $(E_1,E_2)$ respects a sub-constraint $\chi'_v$ for each $v \in V^*$. So $R$ respects $T(G,0) \land \bigwedge_{v \in V^*} \chi'_v$. Finally, Lemma~\ref{lemma:graph_splitting_a_lot_to_connected_equals_far_less_models} shows that $|R| \leq 2^{|E(G)| - |V(G)| - k + 1}$, as required.
\qed
\end{proof}

Theorem~\ref{theorem:main_result} is now a direct consequence of Theorem~\ref{theorem:from_refutation_to_sat_DNNF}, Lemma~\ref{lem:dnnf_lower} and Lemma~\ref{lemma:reduction_to_T(G,0)}
\section{Conclusion}

We have shown that the unsatisfiable Tseitin-formulas with polynomial length of regular resolution refutations are completely determined by the treewidth of the underlying graphs. We did this by giving a connection between lower bounds for regular resolution refutations and size bounds of DNNF representations of Tseitin-formulas. Moreover, we introduced a new two-player game that allowed us to show DNNF lower bounds.

Let us discuss some questions that we think are worth exploring in the future. First, it would be interesting to see if a $2^{\Omega(\tw(G))}$ lower bound for the refutation of Tseitin-formulas can also be shown for general resolution. In that case the length of resolution refutations would essentially be the same as that regular resolution refutations for Tseitin formulas. Note that this is somewhat plausible since other measures like space and width are known to be the same for the two proof systems for these formulas~\cite{GalesiTT20}.

Another question is the relation between knowledge compilation and proof complexity. As far as we are aware, our Theorem~\ref{theorem:from_refutation_to_sat_DNNF} is the first result that connects bounds on DNNF to such in proof complexity. It would be interesting to see if this connection can be strenghtened to other classes of instances, other proof systems, representations from knowledge compilation and measures on proofs and representations, respectively. 

\newpage

\bibliography{biblio}

\end{document}